\documentclass[conference,letterpaper]{IEEEtran}
\IEEEoverridecommandlockouts
\usepackage{cite}
\usepackage{url}
\usepackage{ifthen}
\usepackage{algorithm, algorithmicx, algpseudocode}
\usepackage{graphicx} 

\usepackage{booktabs}
\usepackage{makecell} 
\usepackage[cmex10]{amsmath}
\usepackage{amssymb,amsfonts}
\usepackage[left=1.62cm,right=1.62cm,top=1.85cm]{geometry}
\usepackage{pifont}
\usepackage{amsthm}
\usepackage{mathrsfs}
\def\BibTeX{{\rm B\kern-.05em{\sc i\kern-.025em b}\kern-.08em T\kern-.1667em\lower.7ex\hbox{E}\kern-.125emX}}
\usepackage{tikz}
\usetikzlibrary{automata,positioning,calc}
\usetikzlibrary{intersections}
\usetikzlibrary{decorations.pathreplacing,angles,quotes}

\usepackage{bm}
\usepackage{amscd}

\algnewcommand{\Initialize}[1]{%
  \State \textbf{Initialization:}
  \Statex \hspace*{\algorithmicindent}\parbox[t]{0.8\linewidth}{\raggedright #1}
}
\makeatletter

\theoremstyle{definition}
\newtheorem{theorem}{Theorem}
\newtheorem{definition}{Definition}
\newtheorem{lemma}{Lemma}
\newtheorem{prop}{Proposition}
\newtheorem{cor}{Corollary}
\newtheorem{remark}{Remark}
\newtheorem{eg}{Example}

\newcommand{\norm}[1]{\left\lVert#1\right\rVert}

\newcommand{\argmax}{\operatornamewithlimits{argmax}}

\newcommand{\relmiddle}[1]{\mathrel{}\middle#1\mathrel{}} 

\def\br{\mathbb R}

\def\vE{\mathbb E}

\font\b=cmr10 scaled\magstep4

\def\bigzerou{\smash{\lower1.7ex\hbox{\b 0}}}
\def\bigzerou{\smash{\lower1.7ex\hbox{\b 0}}}

\begin{document}
\title{A Generalized Leakage Interpretation of Alpha-Mutual Information
\thanks{This work was supported by JSPS KAKENHI Grant Number JP23K16886 and JST CRONOS Grant Number JPMJCS25N5.}
}

\author{
\IEEEauthorblockN{Akira Kamatsuka}
\IEEEauthorblockA{Shonan Institute of Technology \\ 
Email: \text{kamatsuka@info.shonan-it.ac.jp}
 }
\and
\IEEEauthorblockN{Takahiro Yoshida}
\IEEEauthorblockA{Nihon University \\ 
Email: \text{yoshida.takahiro@nihon-u.ac.jp}
 } 
}
\maketitle

\begin{abstract}
This paper presents a unified interpretation of $\alpha$-mutual information ($\alpha$-MI) in terms of generalized $g$-leakage.
Specifically, we present a novel interpretation of $\alpha$-MI within an extended framework for quantitative information flow based on adversarial generalized decision problems.
This framework employs the Kolmogorov-Nagumo mean and the $q$-logarithm to characterize adversarial
gain.
Furthermore, we demonstrate that, within this framework, the parameter $\alpha$ can be interpreted as a measure of the adversary's risk aversion.
\end{abstract}

\section{Introduction} \label{sec:intro}
Quantifying the information that a random variable $Y$ conveys about another random variable $X$ is a fundamental problem in information theory and statistical decision-making.
The most widely used measure for this purpose is Shannon mutual information (MI), denoted by $I(X;Y)$ \cite{shannon}.
Since its inception, numerous alternative measures have been proposed across several research disciplines, including information theory, decision theory, and information security, each motivated by distinct operational and conceptual considerations.

Within information theory, a prominent class of generalizations is the family of $\alpha$-mutual information ($\alpha$-MI) \cite{7308959}.
This family encompasses several distinct measures, including Sibson MI \cite{Sibson1969InformationR}, Arimoto MI \cite{arimoto1977}, Augustin--Csisz{\'a}r MI \cite{augusting_phd_thesis,370121}, Hayashi MI \cite{cryptoeprint:2013/440}, and Lapidoth--Pfister MI \cite{e21080778,8231191}, and is typically parameterized by $\alpha \in (0,1)\cup(1,\infty)$.
These measures generalize Shannon MI in closely related ways and play essential roles in diverse operational settings.
For instance, among other applications, $\alpha$-MI characterizes generalized cutoff rates and decoding error exponents in channel coding \cite{370121,e23020199}, error and strong converse exponents in hypothesis testing \cite{6034266,8231191,8007073}, and information leakage in privacy-preserving data publishing \cite{8804205,8943950,9162168,10352344,6266165}.
Extensive research work has established their theoretical properties, operational interpretations, and computing methods often through connections with R{\'e}nyi entropy and divergence \cite{7282554,e21100969,e22050526,Nakiboglu:2019aa,8423117,8849809,e23060702,11078297,10619174,10206941,10619672,8550766,9505206,9064819,11195464,4595361,10858266,10619657,10619200}.

From a decision-theoretic perspective, informativeness has traditionally been quantified using the \textit{value of information} \cite{raiffa1961applied}, which measures the improvements in decision-making performance with respect to $X$ after observing $Y$.
This concept has gained renewed attention in information-theoretic security and quantitative information flow (QIF), where informativeness is interpreted as privacy leakage.
In this context, Issa \textit{et al.} \cite{8943950} introduced \textit{maximal leakage}, defined as the multiplicative increase in an adversary's optimal expected gain after observing $Y$.
This notion was subsequently generalized to \textit{$\alpha$-leakage} and \textit{maximal $\alpha$-leakage} by Liao \textit{et al.} \cite{8804205}.
More broadly, within the QIF framework, Alvim \textit{et al.} \cite{6957119,7536368,ALVIM201932} proposed \textit{$g$-leakage}, which quantifies information leakage under general adversarial gain functions using the concept of vulnerability, the dual notion of entropy.

These developments naturally raise the question of how $\alpha$-MI relates to leakage-based measures of informativeness.
In particular, determining whether and how $\alpha$-MI can be represented within the $g$-leakage framework is of fundamental interest.
An important step in this direction was taken by Liao \textit{et al.} \cite{8804205}, who showed that Arimoto MI coincides with $\alpha$-leakage.
More recently, Kamatsuka and Yoshida \cite{11195284} established that all major variants of $\alpha$-MI admit representations as privacy leakage measures.
Furthermore, Zarrabian and Sadeghi \cite{Zarrabian:2025aa} demonstrated that $\alpha$-leakage--and hence Arimoto MI--can be expressed as a generalized $g$-leakage by replacing the standard expectation with a Kolmogorov--Nagumo (KN) mean \cite{1362262946367354880}.

Despite these advances, existing $g$-leakage representations remain limited in scope.
In particular, a unified $g$-leakage formulation that encompasses $\alpha$-MI beyond the Arimoto case has not yet been established.
Clarifying this relationship is essential for understanding the role of $\alpha$-MI as a general measure of information leakage and for unifying information-theoretic and decision-theoretic perspectives.
The main contributions of this paper are summarized as follows:
\begin{itemize}
\item We introduce a novel definition of generalized conditional vulnerability based on the KN mean (Definition \ref{def:generalized_vulnerability}), employing an approach that differs from existing constructions.
\item We derive a $q$-logarithmic generalization of Gibbs' inequality (Proposition \ref{prop:generalized_Gibbs_inequality}) and use it to relate various forms of R{\'e}nyi conditional entropy to the proposed generalized conditional vulnerability (Theorem \ref{thm:alpha_conditional_entropy_vulnerability}).
\item Building on these results, we establish generalized $g$-leakage representations for $\alpha$-MI beyond the Arimoto case (Corollary \ref{cor:alpha_MI_generalized_leakage}). Furthermore, we demonstrate that the parameter $\alpha$ can be interpreted as a measure of the adversary's risk aversion (Section \ref{ssec:interpretation_alpha}). 
\end{itemize}
These results provide a unified leakage-based interpretation of $\alpha$-MI and deepen the theoretical understanding of its role in information-theoretic security and decision theory.

\section{Preliminaries}\label{sec:preliminaries}
Let $X$ and $Y$ be random variables taking values in finite alphabets $\mathcal{X}$ and $\mathcal{Y}$, respectively, with joint distribution $p_{X,Y} = p_X p_{Y|X}$.
Shannon entropy, conditional entropy, and Shannon mutual information (MI) are defined as follows: 
$H(X):=-\sum_{x}p_{X}(x)\log p_{X}(x)$, $H(X | Y):=-\sum_{x,y}p_{X}(x)p_{Y\mid X}(y | x)\log p_{X\mid Y}(x | y)$, and $I(X; Y):= H(X) - H(X | Y)$, respectively.
Throughout this paper, $\log$ denotes the natural logarithm.
Let $\Delta_{\mathcal{X}}$ denote the set of all probability distributions on $\mathcal{X}$.
For $\alpha>0$ and $p\in\Delta_{\mathcal{X}}$, the \textit{$\alpha$-tilted distribution of $p$} \cite{8804205}, denoted by 
$p^{(\alpha)}\in\Delta_{\mathcal{X}}$, is defined as follows: 
$p^{(\alpha)}(x) := \frac{p(x)^{\alpha}}{\sum_{x\in\mathcal{X}} p(x)^{\alpha}}$. 
For a real-valued random variable $X$ with distribution $p$, its expectation is denoted by $\mathbb{E}^{p}[X]:=\sum_{x}p(x)x$.
Furthermore, let $r_{X|Y}:=\{r_{X|Y}(\cdot|y)\}_{y\in\mathcal{Y}}$ denote a family of conditional distributions on $\mathcal{X}$ given $Y=y$.

\subsection{$\alpha$-Mutual information ($\alpha$-MI)}
We first review R{\' e}nyi entropy, R{\' e}nyi divergence, and several variants of $\alpha$-MI. 
\begin{definition}
Let $\alpha\in (0, 1)\cup (1, \infty)$. Given distributions $p_{X}$ and $q_{X}$, the R{\' e}nyi entropy of order $\alpha$ 
and the R{\' e}nyi divergence of order $\alpha$ between $p_{X}$ and $q_{X}$ are defined as follows:
\begin{align}
H_{\alpha}(p_{X})=H_{\alpha}(X) &:= \frac{1}{1-\alpha} \log \sum_{x} p_{X}(x)^{\alpha}, \label{eq:Renyi_ent} \\
D_{\alpha}(p_{X} || q_{X}) &:= \frac{1}{\alpha-1}\log \sum_{x} p_{X}(x)^{\alpha}q_{X}(x)^{1-\alpha}.
\end{align}
\end{definition}

\begin{definition} \label{def:alpha_MI}
Let $\alpha\in (0, 1)\cup (1, \infty)$ and $(X, Y)\sim p_{X, Y}=p_{X}p_{Y\mid X}$. 
The \textit{Sibson MI, Arimoto MI, Augustin--Csisz{\' a}r MI, Hayashi MI, and Lapidoth--Pfister MI of order $\alpha$}, denoted by $I_{\alpha}^{\text{S}}(X; Y), I_{\alpha}^{\text{A}}(X; Y), I_{\alpha}^{\text{C}}(X; Y), I_{\alpha}^{\text{H}}(X; Y)$, and $I_{\alpha}^{\text{LP}}(X; Y)$, respectively, 
are defined as follows:

\begin{align}
I_{\alpha}^{\text{S}} (X; Y) &:= \min_{q_{Y}} D_{\alpha} (p_{X}p_{Y\mid X} || p_{X}q_{Y}) \label{eq:Sibson_MI} \\ 
I_{\alpha}^{\text{A}}(X; Y) &:= H_{\alpha}(X) - H_{\alpha}^{\text{A}}(X\mid Y), \label{eq:Arimoto_MI} \\
I_{\alpha}^{\text{C}}(X; Y) &:= \min_{q_{Y}}\vE^{p_{X}}\left[D_{\alpha}(p_{Y\mid X}(\cdot \mid X) || q_{Y})\right], \label{eq:AC_MI}\\ 
I_{\alpha}^{\text{H}}(X; Y) &:= H_{\alpha}(X) - H_{\alpha}^{\text{H}}(X\mid Y), \label{eq:Hayshi_MI} \\
I_{\alpha}^{\text{LP}}(X; Y) &:= \min_{q_{X}}\min_{q_{Y}} D_{\alpha}(p_{X}p_{Y\mid X} || q_{X}q_{Y}), \label{eq:LP_MI}
\end{align}
where the minimum in \eqref{eq:Sibson_MI} and \eqref{eq:AC_MI} is taken over all probability distributions on $\mathcal{Y}$,  
and the minimum in \eqref{eq:LP_MI} is taken over all product distributions on $\mathcal{X}\times \mathcal{Y}$. 
The Arimoto conditional entropy of order $\alpha$ \cite{arimoto1977} and the Hayashi conditional entropy of order $\alpha$ \cite{5773033} are defined as follows: 
\begin{align}
H_{\alpha}^{\text{A}}(X | Y)&:= \frac{\alpha}{1-\alpha}\log\sum_{y} \left( \sum_{x}p_{X}(x)^{\alpha}p_{Y\mid X}(y\mid x)^{\alpha} \right)^{\frac{1}{\alpha}}, \label{eq:Arimoto_cond_renyi_ent} \\
H_{\alpha}^{\text{H}}(X | Y)&:= \frac{1}{1-\alpha}\log\sum_{y}p_{Y}(y) \sum_{x}p_{X\mid Y}(x | y)^{\alpha}. \label{eq:Hayashi_cond_renyi_ent}
\end{align}
\end{definition}

\begin{remark}
Each $\alpha$-MI admits a continuous extension to $\alpha=1$ and $\alpha=\infty$.
In particular, $\alpha=1$ recovers Shannon MI $I(X;Y)$.
\end{remark}

Kamatsuka and Yoshida \cite{11195284} demonstrated that $\alpha$-MI variants other than Arimoto MI and Hayashi MI can also be represented as the 
difference between R{\' e}nyi entropy and conditional R{\' e}nyi entropy of the form $I_{\alpha}^{(\cdot)}(X; Y):= H_{\alpha}(X) - H_{\alpha}^{(\cdot)}(X | Y)$.

\begin{prop}[\text{\cite[Thm 1]{11195284}}] \label{prop:alpha_MI_diff}
For $\alpha\in (0, 1)\cup (1, \infty)$, 
\begin{align}
I_{\alpha}^{\text{S}}(X; Y) &= H_{\frac{1}{\alpha}}(X) - H_{\alpha}^{\text{S}}(X\mid Y), \label{eq:Sibson_diff_expression}\\ 
I_{\alpha}^{\text{C}}(X; Y) &= H(X) - H_{\alpha}^{\text{C}}(X\mid Y), \label{eq:AC_diff_expression}
\end{align}
where 
\begin{align}
&H_{\alpha}^{\text{S}}(X|Y) := \min_{r_{X\mid Y}}\frac{\alpha}{1-\alpha}\notag \\ 
&\times \log \sum_{x,y}p_{X_{\frac{1}{\alpha}}}(x)p_{Y\mid X}(y | x)r_{X\mid Y}(x | y)^{1-\frac{1}{\alpha}}, \label{eq:Sibson_cond_renyi_ent} \\ 
&H_{\alpha}^{\text{C}}(X|Y) := \min_{r_{X\mid Y}}\frac{\alpha}{1-\alpha} \notag \\ 
&\times \sum_{x}p_{X}(x)\log \sum_{y}p_{Y\mid X}(y|x)r_{X\mid Y}(x|y)^{1-\frac{1}{\alpha}}, \label{eq:AC_cond_renyi_ent}
\end{align}
and $p_{X_{\frac{1}{\alpha}}} = p_{X}^{(\frac{1}{\alpha})}$, the $\frac{1}{\alpha}$-tilted distribution of $p_{X}$. 

For $\alpha\in (1/2, 1) \cup (1, \infty)$, 
\begin{align}
I_{\alpha}^{\text{LP}}(X; Y) &= H_{\frac{\alpha}{2\alpha-1}}(X) - H_{\alpha}^{\text{LP}}(X\mid Y), \label{eq:LP_cond_renyi_ent}
\end{align}
where 
\begin{align}
&H_{\alpha}^{\text{LP}}(X|Y) := \min_{r_{X\mid Y}}\frac{2\alpha-1}{1-\alpha} \notag \\ 
&\times \log \sum_{x}p_{X_{\frac{\alpha}{2\alpha-1}}}(x)\left( \sum_{y}p_{Y\mid X}(y|x)r_{X\mid Y}(x|y)^{1-\frac{1}{\alpha}} \right)^{\frac{\alpha}{2\alpha-1}}
\end{align}
and $p_{X_{\frac{\alpha}{2\alpha-1}}} = p_{X}^{(\frac{\alpha}{2\alpha-1})}$, the $\frac{\alpha}{2\alpha-1}$-tilted distribution of $p_{X}$. 
\end{prop}

\subsection{$q$-Logarithm and Kolmogorov--Nagumo (KN) Mean} \label{ssec:q_logarithm_KN_mean}
We next introduce the $q$-logarithm, the KN mean, and a $q$-logarithmic generalization of Gibbs' inequality.

\begin{definition}[$q$-logarithm and $q$-exponential]
Let $q\in (-\infty, \infty)$. 
The \textit{$q$-logarithm} and \textit{$q$-exponential} are defined as follows:
\begin{align}
\ln_{q}x &:= 
\begin{cases}
\log x, & q=1, \\ 
\frac{x^{1-q}-1}{1-q}, & q\neq 1, 
\end{cases} \\ 
\exp_{q}\{x\} &:= 
\begin{cases}
\exp\{x\}, & q=1, \\ 
[1+(1-q)x]^{\frac{1}{1-q}}, & q\neq 1.
\end{cases}
\end{align}
\end{definition}
\begin{remark}
The functions $\ln_q(\cdot)$ and $\exp_q\{\cdot\}$ are strictly increasing and mutually inverse.
\end{remark}

\begin{prop}[Generalized Gibbs' inequality] \label{prop:generalized_Gibbs_inequality}
Let $q\in (0, 1)\cup (1, \infty)$. Then, the following holds:
\begin{align}
\max_{r\in \Delta_{\mathcal{X}}} \sum_{x}p_{X}(x)\ln_{q}r(x) 
&= \ln_{q} \norm{p_{X}}_{\frac{1}{q}}^{\frac{1}{1-q}}, \label{eq:generalized_Gibbs_inequality}
\end{align}
where the maximum in \eqref{eq:generalized_Gibbs_inequality} is attained at $r = p_{X}^{({1}/{q})}$, the $1/q$-tilted distribution of $p_{X}$. 
\end{prop}
\begin{proof}
See Appendix \ref{proof:generalized_Gibbs_inequality}. 
\end{proof}
\begin{remark}
For $q\to 1$, Eq.~\eqref{eq:generalized_Gibbs_inequality} reduces to the classical Gibbs' inequality \cite{bhl76948} (also known as the Shannon--Kolmogorov information inequality): 
For any $r\in \Delta_{\mathcal{X}}$, 
\begin{align}
\sum_{x}p_{X}(x)\log r(x) \leq \sum_{x}p_{X}(x)\log p_{X}(x), \label{eq:Gibbs_inequality}
\end{align}
with equality if and only if $r=p_{X}$.
\end{remark}

\begin{definition}[Kolmogorov--Nagumo (KN) mean]
Let $I\subseteq \br$ and $\varphi\colon I\to \br$ be a strictly monotonic and continuous function. 
Given a distribution $p\in \Delta_{\mathcal{X}}$, the \textit{KN mean} (also known as the \textit{quasi-arithmetic mean}) \textit{of $X$} is defined as follows:
\begin{align}
\mathbb{M}^{p}_{\varphi}[X] &:= \varphi^{-1} \left( \vE^{p}\left[\varphi(X)\right] \right) \\
&= \varphi^{-1}\left( \sum_{x}p(x)\varphi(x) \right),
\end{align}
where $\varphi^{-1}$ denotes the inverse function of $\varphi$.
\end{definition}

\begin{eg}
Examples of the KN mean include the following: 
\begin{itemize}
\item If $\varphi(t)=at + b, a\neq 0$, then $\mathbb{M}^{p}_{\varphi}[X] = \vE^{p}[X]$ (expectation).
\item If $\varphi(t)=\log t$, then $\mathbb{M}^{p}_{\varphi}[X] = \mathbb{G}^{p}[X]:=\prod_{x}x^{p(x)}$ (geometric mean). 
\item If $\varphi(t)=\ln_{q} t$, then $\mathbb{M}^{p}_{\varphi}[X] = \mathbb{H}^{p}_{1-q}[X] := \{\vE^{p}[X^{1-q}]\}^{\frac{1}{1-q}}$ (H{\" o}lder mean of order $1-q$). 
\end{itemize}
\end{eg}

\section{Generalized Vulnerability and Generalized $g$-Leakage}\label{sec:generalized_vulnerarbility_Gibbs_inequality}
This section reviews the notion of generalized vulnerability introduced by Zarrabian and Sadeghi \cite{Zarrabian:2025aa} and introduces a novel formulation of conditional vulnerability, based on a construction distinct from that of \cite{Zarrabian:2025aa}.
Based on these notions, we define a generalized $g$-leakage that serves as a unifying framework for $\alpha$-MI. 

Suppose that an adversary takes an action $A$ regarding the original data $X$ after observing the disclosed information $Y$.
The adversary employs a decision rule $\delta\colon\mathcal{Y}\to\mathcal{A}$ in conjunction with a gain function $g\colon\mathcal{X}\times\mathcal{A}\to\mathbb{R}$, where $\mathcal{A}$ denotes the action space.
When $\mathcal{A}=\Delta_{\mathcal{X}}$, we denote $\delta(y)=r_{X|Y}(\cdot|y)\in\Delta_{\mathcal{X}}$.

\begin{definition}[Generalized vulnerability] \label{def:generalized_vulnerability}
Let $(X,Y)\sim p_X\times p_{Y|X}$ and let $g(x,a)$ be a gain function.
Let $I\subseteq \br$ and let $\varphi,\psi\colon I\to\mathbb{R}$ be continuous and strictly monotone functions.
The \textit{generalized prior vulnerability} and \textit{generalized conditional vulnerability} are defined as follows:
\begin{align}
V^{p}_{\varphi, g}(X) &= \max_{a} \mathbb{M}^{p}_{\varphi}[g(X, a)], \\ 
V^{p}_{\varphi, \psi, g}(X\mid Y) &= \max_{\delta} \mathbb{M}^{p}_{\varphi}\left[\mathbb{M}^{p_{Y\mid X}(\cdot \mid X)}_{\psi}\left[g(X, \delta(Y))\relmiddle{|}X\right]\right]. \label{eq:generalized_vulnerability}
\end{align}
\end{definition}
\begin{remark}
By replacing maximization with minimization, these definitions yield generalized entropy functionals based on a loss function $\ell(x, a)$ rather than a gain function.
We denote these by $H^{p}_{\varphi,\ell}(X)$ and $H^{p}_{\varphi,\psi,\ell}(X|Y)$.
\end{remark}

\begin{remark}
Zarrabian and Sadeghi \cite[Def 10]{Zarrabian:2025aa} introduced an alternative quantity, \textit{generalized average posterior vulnerability}, defined as 
\begin{align}
\hat{V}^{p}_{\psi, \varphi, g}(X\mid Y) &= \mathbb{M}^{p_{Y}}_{\psi}\left[\max_{a} \mathbb{M}^{p_{X\mid Y}(\cdot \mid Y)}_{\varphi}\left[g(X, a)\relmiddle{|} Y\right]\right], 
\end{align}
where $p_Y$ and $p_{X\mid Y}$ are the marginal and conditional distributions induced by $p_X\times p_{Y|X}$.
While our proposal generalizes the Bayes expected gain via KN means, \cite{Zarrabian:2025aa} generalizes the posterior-optimal expected gain.
In general, $V^{p}_{\varphi,\psi,g}(X | Y)$ and $\hat{V}^{p}_{\psi,\varphi,g}(X | Y)$ do not coincide, though, they agree when $\varphi=\psi$. 
\end{remark}

\begin{prop}
For $\varphi=\psi$, the following holds:
\begin{align}
&V^{p}_{\varphi, \varphi, g}(X\mid Y) \notag \\ 
&=  \begin{cases}
\varphi^{-1}\left( \sum_{y}p_{Y}(y)\max_{a}\sum_{x}p_{X\mid Y}(x\mid y)\varphi(g(x, a)) \right), \notag \\ 
\qquad \qquad \qquad \qquad \qquad \qquad \qquad \qquad \qquad  \text{$\varphi$: increasing}, \\
\varphi^{-1}\left( \sum_{y}p_{Y}(y)\min_{a}\sum_{x}p_{X\mid Y}(x\mid y)\varphi(g(x, a)) \right), \notag \\ 
\qquad \qquad \qquad \qquad \qquad \qquad \qquad \qquad \qquad   \text{$\varphi$: decreasing}, 
\end{cases}
\\
&= \hat{V}^{p}_{\varphi, \varphi, g}(X\mid Y)
\end{align}
\end{prop}
\begin{proof}
The first equality follows from \cite[Thm 2.7]{ghosh2007introduction}; The second equality follows from \cite[Remark 3]{Zarrabian:2025aa}.
\end{proof}

\begin{remark}\label{rem:posterior_optimality_phi}
When $\varphi=\psi$ is increasing, the optimal decision rule maximizes the posterior expected gain for each realization $Y=y$ with respect to the \textit{transformed gain function} $\varphi\circ g$.
Specifically, the optimal action $\delta^{*}(y)$ satisfies
\begin{align}
\delta^{*}(y)  = \argmax_{a}\sum_{x} p_{X\mid Y}(x\mid y)\,\varphi\bigl(g(x,a)\bigr).
\end{align}
\end{remark}

Next, we define the two gain functions that will be used to represent R{\' e}nyi-type entropies and $\alpha$-MI.

\begin{definition} \label{def:soft_power_score}
Let $\mathcal{A} = \Delta_{\mathcal{X}}$ and $\alpha\in (0, 1)\cup (1, \infty)$. 
The \textit{soft $0$-$1$ score} and \textit{power score} \cite{Selten:1998aa} (also known as \textit{Tsallis score}), denoted by $g_{\text{$0$-$1$}}(x, r), f_{\alpha, \text{PW}}(x, r)$, are defined as follows:
\begin{align}
g_{\text{$0$-$1$}}(x, r) &:= r(x), \\ 
f_{\alpha, \text{PW}}(x, r) &:= \alpha r(x)^{\alpha-1} + (1-\alpha)\sum_{x}r(x)^{\alpha}.
\end{align}
\end{definition}
\begin{remark} \label{rem:soft_power_score}
The soft $0$-$1$ score $g_{\text{$0$-$1$}}(x, r)$ is dual to the soft $0$-$1$ loss $\ell_{\text{$0$-$1$}}(x, r):=1-r(x)$ introduced in \cite[Def 3]{8804205}.
The power score $f_{\alpha}(x, r)$ is originally given for $\alpha>1$ \cite[Section 2.7]{Selten:1998aa}; as shown below, it behaves as a gain for $\alpha>1$ and as a loss for $0<\alpha<1$.
\end{remark}
\begin{prop} \label{prop:power_score}
\begin{align}
\begin{cases}
\max_{r} \sum_{x}p_{X}(x)f_{\alpha, \text{PW}}(x, r) = \norm{p_{X}}_{\alpha}^{\alpha}, & \alpha>1, \\ 
\min_{r} \sum_{x}p_{X}(x)f_{\alpha, \text{PW}}(x, r) = \norm{p_{X}}_{\alpha}^{\alpha}, & 0<\alpha<1.
\end{cases}
\end{align}
\end{prop}
\begin{proof}
The case $\alpha>1$ is shown in \cite[Section 2.7]{Selten:1998aa}; the case $0 < \alpha < 1$ follows by the same argument.
\end{proof}

Finally, we introduce the following generalized $g$-leakage.
\begin{definition}[Generalized $g$-leakage]
Let $(X, Y)\sim p\times p_{Y\mid X}$. Then, the \textit{generalized multiplicative $g$-leakage}, denoted by $\mathcal{L}_{\varphi, \psi, g}^{\times, p}(X\to Y)$ is defined as follows:
\begin{align}
\mathcal{L}_{\varphi, \psi, g}^{\times, p}(X\to Y) &:= \log \frac{V^{p}_{\varphi, \psi, g}(X\mid Y)}{V^{p}_{\varphi, g}(X)}.
\end{align}
\end{definition}
\begin{remark}
When a loss function $\ell(x,a)$ is used in place of a gain function, we denote it as 
$\mathcal{L}^{\times,p}_{\varphi,\psi,\ell}(X\to Y):= \log\frac{H^{p}_{\varphi,\ell}(X)}{H^{p}_{\varphi,\psi,\ell}(X\mid Y)}$ with slight abuse of notation.
\end{remark}

\section{Main Results}\label{sec:main_result}

In this section, we establish that each variant of (conditional) R{\' e}nyi entropy $H_{\alpha}^{(\cdot)}(X | Y)$ appearing in Definition \ref{def:alpha_MI} and Proposition \ref{prop:alpha_MI_diff} admits a representation in terms of generalized vulnerability.
Building on these results, we derive generalized $g$-leakage expressions for $\alpha$-MI and  provide a decision-theoretic interpretation of $\alpha$ as a measure of adversarial risk aversion.

\subsection{Generalized $g$-Leakage expressions for $\alpha$-MI}
\begin{theorem} \label{thm:alpha_conditional_entropy_vulnerability}
Let $\mathcal{A} = \Delta_{\mathcal{X}}$ and let $g_{\text{$0$-$1$}}(x, r)$ and $f_{\alpha, \text{PW}}(x, r)$ denote the soft $0$-$1$ score and the power score defined in Definition \ref{def:soft_power_score}, respectively. 
Then, for $\alpha\in (0, 1)\cup (1, \infty)$, the following hold:
\begin{align}
V_{\log t, g_{\text{$0$-$1$}}}^{p_{X}}(X) &= \exp\{-H(X)\}, \label{eq:vulnerability_entropy} \\
V_{\log t, \log t, g_{\text{$0$-$1$}}}^{p_{X}}(X\mid Y) &= \exp\{-H(X\mid Y)\}, \label{eq:vulnerability_conditional_entropy} \\
V_{\ln_{\frac{1}{\alpha}}t, \ln_{\frac{1}{\alpha}}t,  g_{\text{$0$-$1$}}}^{p_{X}}(X) &= \exp\{-H_{\alpha}(X)\}, \label{eq:vulnerability_Renyi_entropy_01} \\ 
V_{\ln_{\frac{1}{\alpha}}t, \ln_{\frac{1}{\alpha}}t,  g_{\text{$0$-$1$}}}^{p_{X}}(X\mid Y) &= \exp\{-H_{\alpha}^{\text{A}}(X\mid Y)\}, \label{eq:vulnerability_Arimoto_entropy} \\ 
V_{\ln_{\frac{1}{\alpha}}t, \ln_{\frac{1}{\alpha}}t,  g_{\text{$0$-$1$}}}^{p_{X_{\frac{1}{\alpha}}}}(X\mid Y) &= \exp\{-H_{\alpha}^{\text{S}}(X\mid Y)\}, \label{eq:vulnerability_Sibson_entropy} \\ 
V_{\log t, \ln_{\frac{1}{\alpha}}t, g_{\text{$0$-$1$}}}^{p_{X}}(X\mid Y) &= \exp\{-H_{\alpha}^{\text{C}}(X\mid Y)\}, \label{eq:vulnerability_AC_entropy} \\ 
H_{\ln_{\alpha}, f_{\alpha, \text{PW}}^{\frac{1}{1-\alpha}}}^{p_{X}}(X) &= \exp\{H_{\alpha}(X)\}, \label{eq:vulnerability_Renyi_entropy_02} \\ 
H_{\ln_{\alpha}t, \ln_{\alpha}t, f_{\alpha, \text{PW}}^{\frac{1}{1-\alpha}}}^{p_{X}}(X\mid Y) &= \exp\{H_{\alpha}^{\text{H}}(X\mid Y)\}. \label{eq:vulnerability_Hayashi_entropy}
\end{align}
For $\alpha\in (1/2, 1)\cup (1, \infty)$, 
\begin{align}
V_{\ln_{\frac{\alpha}{2\alpha-1}}t, \ln_{\frac{1}{\alpha}}t, g_{\text{$0$-$1$}}}^{p_{X_{\frac{\alpha}{2\alpha-1}}}}(X\mid Y) = \exp\{-H_{\alpha}^{\text{LP}}(X\mid Y)\}. \label{eq:vulnerability_LP_entropy}
\end{align}
\end{theorem}
\begin{proof}
See Appendix \ref{proof:alpha_conditional_entropy_vulnerability}. 
\end{proof}

From Theorem \ref{thm:alpha_conditional_entropy_vulnerability}, we immediately obtain the following generalized $g$-leakage characterizations of $\alpha$-MI.

\begin{cor} \label{cor:alpha_MI_generalized_leakage}
Let $(X, Y)\sim p_{X}p_{Y\mid X}$. For $\alpha\in (0, 1)\cup (1, \infty)$, 
\begin{align}
I(X; Y) &= \mathcal{L}_{\log t, \log t, g_{\text{$0$-$1$}}}^{\times, p_{X}}(X\to Y), \\ 
I_{\alpha}^{\text{A}}(X; Y) &= \mathcal{L}_{\ln_{\frac{1}{\alpha}}t, \ln_{\frac{1}{\alpha}}t, g_{\text{$0$-$1$}}}^{\times, p_{X}}(X\to Y), \label{eq:Arimoto_MI_leakage} \\ 
I_{\alpha}^{\text{S}}(X; Y) &= \mathcal{L}_{\ln_{\frac{1}{\alpha}}t, \ln_{\frac{1}{\alpha}}t, g_{\text{$0$-$1$}}}^{\times, p_{X_{\frac{1}{\alpha}}}}(X\to Y), \\ 
I_{\alpha}^{\text{C}}(X; Y) &= \mathcal{L}_{\log t, \ln_{\frac{1}{\alpha}}t, g_{\text{$0$-$1$}}}^{\times, p_{X}}(X\to Y), \\ 
I_{\alpha}^{\text{H}}(X; Y) &= \mathcal{L}_{\ln_{\alpha}t, \ln_{\alpha}t, f_{\alpha, \text{PW}}^{\frac{1}{1-\alpha}}}^{\times, p_{X}}(X\to Y). 
\end{align}
For $\alpha\in (1/2, 1)\cup (1, \infty)$, 
\begin{align}
I_{\alpha}^{\text{LP}}(X; Y) &= \mathcal{L}_{\ln_{\frac{\alpha}{2\alpha-1}}t, \ln_{\frac{1}{\alpha}t}, g_{\text{$0$-$1$}}}^{\times, p_{X_{\frac{\alpha}{2\alpha-1}}}}(X\to Y). 
\end{align}
\end{cor}

\begin{remark}\label{rem:comparison_alpha_MI_leakage}
First, Corollary \ref{cor:alpha_MI_generalized_leakage} shows that the distinctions among various
$\alpha$-MI measures are fully captured by the tuple $(p, \varphi, \psi, g)$.
Different definitions correspond to specific methods of aggregating the adversarial gain via the KN mean and to different choices of prior distribution. 
In particular,  the generalized $g$-leakage representation of Arimoto MI coincides with the result of Zarrabian and Sadeghi\cite[Prop 2]{Zarrabian:2025aa}.
Second, except for Hayashi MI, all generalized leakage representations employ the same gain function, namely the soft $0$-$1$ score $g_{\text{$0$-$1$}}(x, r)$.
Third, in the generalized $g$-leakage representations derived in Corollary~\ref{cor:alpha_MI_generalized_leakage}, Shannon MI, Arimoto MI, and Sibson MI correspond to the special case $\varphi=\psi$. 
As discussed in Remark \ref{rem:posterior_optimality_phi}, this implies that the adversary's optimal decision reduces to maximizing posterior expected gain with respect to a transformed gain function $\varphi\circ g$.
Arimoto MI employs the original prior $p_X$, whereas Sibson MI uses the $1/\alpha$-tilted prior $p_{X_{\frac{1}{\alpha}}}=p_X^{(1/\alpha)}$. 
This distinction admits a natural interpretation in terms of the adversary's prior belief. 
For $0<\alpha<1$, the tilted distribution $p_X^{(1/\alpha)}$ overemphasizes symbols with large probability under $p_X$ while effectively ignoring low-probability symbols, corresponding to an adversary that focuses on highly likely events. 
For $\alpha>1$, the tilting operation downweights high-probability symbols and upweights low-probability ones, yielding a belief closer to the uniform distribution and reflecting a more balanced or uncertainty-averse prior held by the adversary. Similar observations appear in \cite{10619672}, \cite{11195284}, and \cite{9761766}.
\end{remark}

\subsection{An Interpretation of $\alpha$ as the Adversary's Risk Aversion}
\label{ssec:interpretation_alpha}

We now provide a utility-theoretic interpretation of the parameter $\alpha$ from the perspective of the adversary's risk attitude.
We focus on the case of Arimoto MI, whose generalized $g$-leakage representation corresponds to the choice
$(p,\varphi,\psi,g) = (p_X,\ln_{1/\alpha} t,\ln_{1/\alpha} t,g_{\text{$0$-$1$}})$.
In this setting, the transformed gain function is given by
\begin{align}
g_{\alpha}(x,r) &:= \varphi\!\left(g_{0\text{-}1}(x,r)\right) = \ln_{1/\alpha} r(x),
\end{align}
where $r(x)$ represents the estimated probability assigned by the adversary to the true symbol $x$.
Figure \ref{fig:600_dpi_comparison_alpha_MI} illustrates the shape of the transformed gain function $g_{\alpha}(x,r)$ as a function of $r(x)$ for different values of $\alpha$.
As $\alpha$ decreases, the transformed gain function increasingly penalizes small values of $r(x)$, assigning disproportionately low gain when the probability of correctly guessing $x$ is small. 
In contrast, as $\alpha$ increases, the gain function becomes closer to linear in $r(x)$. In the limiting case $\alpha=\infty$, the transformed gain approaches $r(x)-1$, corresponding to a risk-neutral adversary.
This observation admits a natural interpretation in terms of expected utility theory (see, e.g., \cite{alma990612587980206881}).
In particular, the curvature of the gain function $g_{\alpha}(x,r)$ quantifies the adversary's attitude toward risk. Using the Arrow--Pratt measure of absolute risk aversion $A_g(t)$, defined for a twice-differentiable gain function $g$, 
\begin{align}
A_g(t)&:=-\,\frac{g^{\prime \prime}(t)}{g^{\prime}(t)},
\end{align}
a direct calculation yields
\begin{align}
A_{g_{\alpha}}(r(x)) &= \frac{1}{\alpha\, r(x)}.
\end{align}
Since $A_{g_{\alpha}}(r(x))$ is monotonically decreasing in $\alpha$ for each $r(x)\in (0,1]$, smaller values of $\alpha$ correspond to more risk-averse adversaries, whereas larger values of $\alpha$ correspond to less risk-averse (and hence more risk-neutral) adversaries.
Therefore, the parameter $\alpha$ in Arimoto MI admits a clear utility-theoretic interpretation as a measure of the adversary's degree of risk aversion.

\begin{figure}[t]
\centering
\includegraphics[width=0.5\textwidth,keepaspectratio, clip]{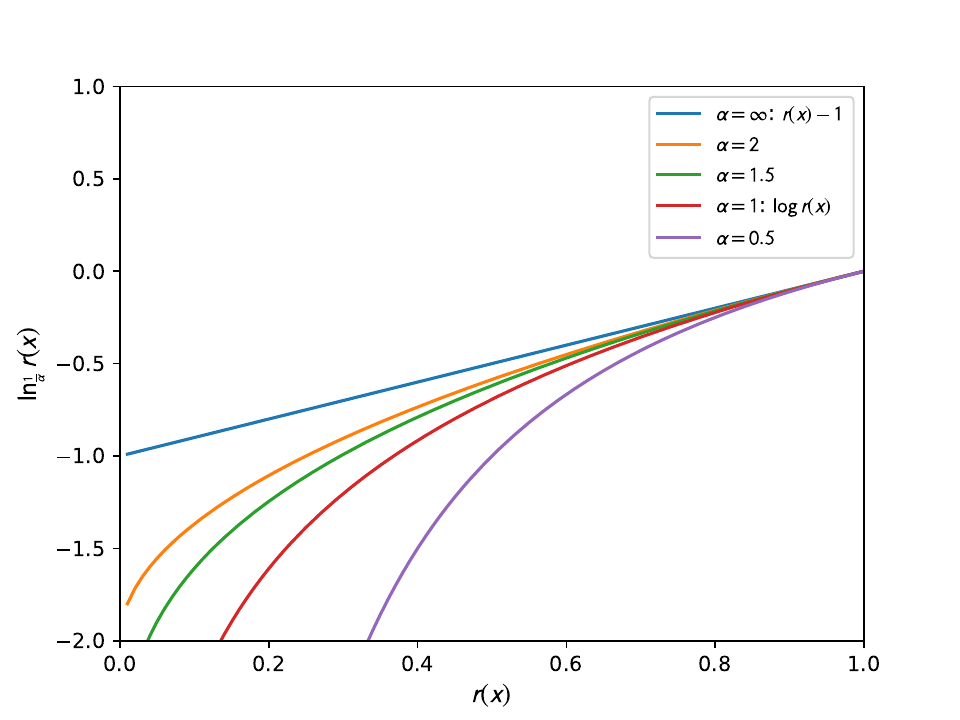}
\caption{Plot of $g_{\alpha}(x, r)=\ln_{1/\alpha} r(x)$}
\label{fig:600_dpi_comparison_alpha_MI}
\end{figure}

\section{Conclusion}\label{sec:conclusion}
This paper established novel decision-theoretic interpretations of $\alpha$-MI as generalized $g$-leakage measures within an extended QIF framework based on KN means and the $q$-logarithm.
By introducing a new notion of conditional vulnerability, formulated through a decision rule-based approach distinct from existing definitions, and deriving a $q$-logarithmic generalization of Gibbs' inequality, 
we enabled unified generalized $g$-leakage representations of various $\alpha$-MI, including forms beyond Arimoto MI. 
The resulting framework clarifies that the differences among $\alpha$-MI measures stem from the choice of aggregation functions and prior distributions in the underlying leakage representations, 
thereby providing a coherent structural understanding of their relationships.
Furthermore, focusing on Arimoto MI, we showed that the parameter $\alpha$ admits a clear utility-theoretic interpretation as the adversary’s degree of risk aversion, quantified via the curvature of the transformed gain function.


\appendices

\section{Proof of Proposition \ref{prop:generalized_Gibbs_inequality}}\label{proof:generalized_Gibbs_inequality}
To prove Proposition \ref{prop:generalized_Gibbs_inequality}, we first recall the following variant of H{\" older}'s inequality.

\begin{lemma}[Reverse H{\" o}lder's inequality \text{\cite{2946733}}] \label{lem:Holder_rev_inequality}
Let $a_{i}\geq 0, b_{i}\geq 0, i=1,\dots, n$ and $p\in (0, 1)\cup (1, \infty)$. Then, the following holds:
\begin{align}
\begin{cases}
\sum_{i=1}^{n} a_{i}b_{i} \geq \left( \sum_{i=1}^{n} a_{i}^{\frac{1}{p}} \right)^{p} \left( \sum_{i=1}^{n} b_{i}^{\frac{1}{1-p}} \right)^{1-p}, & p > 1, \\ 
\sum_{i=1}^{n} a_{i}b_{i} \leq \left( \sum_{i=1}^{n} a_{i}^{\frac{1}{p}} \right)^{p} \left( \sum_{i=1}^{n} b_{i}^{\frac{1}{1-p}} \right)^{1-p}, & 0 < p < 1, 
\end{cases}
\label{eq:Holder_rev_inequality}
\end{align}
with equality if and only if there exists a constant $c$ such that for all $i=1,\dots, n$, $a_{i} = cb_{i}^{\frac{p}{1-p}}$.
\end{lemma}

Using Lemma \ref{lem:Holder_rev_inequality}, we prove Proposition \ref{prop:generalized_Gibbs_inequality} as follows.
\begin{align}
&\sum_{x}p_{X}(x)\ln_{q}r(x) \notag \\ 
&= \frac{1}{1-q} \left( \sum_{x}p_{X}(x)r(x)^{1-q}-1 \right) \\ 
&\overset{(a)}{\leq} \frac{1}{1-q} \left( \left( \sum_{x} p_{X}(x)^{\frac{1}{q}} \right)^{q}\left( \sum_{x}r(x)^{(1-q)\cdot \frac{1}{1-q}} \right)^{1-q} -1\right) \label{eq:proof_generalized_Gibbs_inequality} \\
&= \frac{1}{1-q} \left( \norm{p_{X}}_{\frac{1}{q}} - 1\right) = \ln_{q} \norm{p_{X}}_{\frac{1}{q}}^{\frac{1}{1-q}},
\end{align}
where $(a)$ and its equality condition follow from Lemma \ref{lem:Holder_rev_inequality}. 

\section{Proof of Theorem \ref{thm:alpha_conditional_entropy_vulnerability}} \label{proof:alpha_conditional_entropy_vulnerability}

We now prove Theorem \ref{thm:alpha_conditional_entropy_vulnerability} step by step. 
\begin{align} 
&V_{\log t, g_{\text{$0$-$1$}}}^{p_{X}}(X) = \max_{r} \exp\left\{\sum_{x}p_{X}(x)\log r(x) \right\} \\
&= \exp\left\{\max_{r}\sum_{x}p_{X}(x)\log r(x) \right\} \overset{(a)}{=} \exp\{-H(X)\}, \\ 
&V_{\log t, \log t, g_{\text{$0$-$1$}}}^{p_{X}}(X\mid Y) \notag \\ 
&= \max_{r_{X\mid Y}} \exp\left\{\sum_{x,y}p_{X}(x)p_{Y\mid X}(y\mid x)\log r_{X\mid Y}(x\mid y) \right\} \\
&= \exp\left\{\sum_{y}p_{Y}(y)\max_{r_{X\mid Y}(\cdot \mid y)}p_{X\mid Y}(x\mid y)\log r_{X\mid Y}(x\mid y) \right\} \notag \\ 
&\overset{(b)}{=} \exp\{-H(X\mid Y)\}, 
\end{align}
where $(a)$ and $(b)$ follow from Gibbs' inequality \eqref{eq:Gibbs_inequality} and definitions of Shannon entropy and condtional entropy. 
\begin{align}
&V_{\ln_{\frac{1}{\alpha}}t, g_{\text{$0$-$1$}}}^{p_{X}}(X) 
= \max_{r} \exp_{\frac{1}{\alpha}}\left\{\sum_{x}p_{X}(x)\ln_{\frac{1}{\alpha}}r(x)\sum\right\} \\ 
&= \exp_{\frac{1}{\alpha}}\left\{\max_{r}\sum_{x}p_{X}(x)\ln_{\frac{1}{\alpha}}r(x)\right\} \\ 
&\overset{(c)}{=} \exp_{\frac{1}{\alpha}}\left\{\norm{p_{X}}_{\alpha}^{\frac{\alpha}{\alpha-1}}\right\} \overset{(d)}{=} \exp\left\{H_{\alpha}(X)\right\}, \\ 
&V_{\ln_{\frac{1}{\alpha}}t, \ln_{\frac{1}{\alpha}}t, g_{\text{$0$-$1$}}}^{p_{X}}(X\mid Y) \notag \\
&= \max_{r_{X\mid Y}} \exp_{\frac{1}{\alpha}}\left\{\sum_{x,y}p_{X}(x)p_{Y\mid X}(y\mid x)\ln_{\frac{1}{\alpha}} r_{X\mid Y}(x\mid y) \right\} \\ 
&= \exp_{\frac{1}{\alpha}}\left\{\max_{r_{X\mid Y}} \sum_{x,y}p_{X}(x)p_{Y\mid X}(y\mid x)\ln_{\frac{1}{\alpha}} r_{X\mid Y}(x\mid y) \right\} \\ 
&= \exp_{\frac{1}{\alpha}}\left\{\sum_{y}p_{Y}(y)\max_{r_{X\mid Y}(\cdot\mid y)}\sum_{x}p_{X\mid Y}(x | y)\ln_{\frac{1}{\alpha}} r_{X\mid Y}(x | y) \right\} \\ 
&\overset{(e)}{=} \exp_{\frac{1}{\alpha}}\left\{\sum_{y}p_{Y}(y)\ln_{\frac{1}{\alpha}} \norm{p_{X\mid Y}(\cdot \mid y)}_{\alpha}^{\frac{\alpha}{\alpha-1}} \right\} \\ 
&= \left( \sum_{y}p_{Y}(y) \norm{p_{X\mid Y}(\cdot \mid y)}_{\alpha} \right)^{\frac{\alpha}{\alpha-1}} \overset{(f)}{=} \exp\left\{-H_{\alpha}^{\text{A}}(X\mid Y) \right\},
\end{align}
where 
\begin{itemize}
\item $(c)$ and $(e)$ follow from Proposition \ref{prop:generalized_Gibbs_inequality} by substituting $q={1}/{\alpha}$.
\item $(d)$ and $(f)$ follow from Eqs.~\eqref{eq:Renyi_ent} and \eqref{eq:Arimoto_cond_renyi_ent}, respectively.
\end{itemize}
Eq.~\eqref{eq:vulnerability_Sibson_entropy} follows from Eq.~\eqref{eq:vulnerability_Arimoto_entropy} by replacing $p_X$ with $p_{X_{\frac{1}{\alpha}}}$.
\begin{align}
&V_{\log t, \ln_{q}t, g_{\text{$0$-$1$}}}^{p_{X}}(X\mid Y) \notag \\
&= \max_{r_{X\mid Y}} \exp\Big\{\sum_{x}p_{X}(x) \notag \\ 
&\qquad \qquad  \times \log \Big(\exp_{q}\Big\{\sum_{y}p_{Y\mid X}(y\mid x)\ln_{q} r_{X\mid Y}(x\mid y) \Big\} \Big)\Big\} \\ 
&= \exp\Big\{\max_{r_{X\mid Y}}\Big( \sum_{x}p_{X}(x) \notag \\ 
&\qquad \qquad \times \log \sum_{y}p_{Y\mid X}(y\mid x)r_{X\mid Y}(x | y)^{1-\frac{1}{\alpha}} \Big)^{1-\frac{1}{\alpha}}\Big\} \\
&\overset{(f)}{=} \exp\left\{-H_{\alpha}^{\text{C}}(X\mid Y) \right\},
\end{align}
where $(f)$ follows from Eq.~\eqref{eq:AC_cond_renyi_ent}.

Define the following loss function\footnote{Notably, from Remark \ref{rem:soft_power_score} and the fact that 
$t\mapsto  t^{\frac{1}{1-\alpha}}$ is increasing for $\alpha\in (0, 1)$; decresing for $\alpha\in (1, \infty)$, the function $\tilde{\ell}_{\alpha}(x, r)$ can be seen as a loss function.}: 
\begin{align}
\tilde{\ell}_{\alpha}(x, r) &:= f_{\alpha, \text{PW}}(x, r)^{\frac{1}{1-\alpha}}.
\end{align}
Then, we obtain the following: 
\begin{align}
&H_{\ln_{\alpha}t, \tilde{\ell}_{\alpha}}^{p_{X}}(X) = \min_{r} \exp_{\alpha}\left\{\sum_{x}p_{X}(x)\ln_{\alpha} \tilde{\ell}_{\alpha}(x, r) \right\} \\
&= \exp_{\alpha}\left\{\min_{r} \sum_{x}p_{X}(x) \tilde{\ell}_{\alpha}(x, r)\right\} \\ 
&= \begin{cases}
\left[ \max_{r} \sum_{x}p_{X}(x) f_{\alpha, \text{PW}}(x, r) \right]^{\frac{1}{1-\alpha}}, & \alpha>1, \\
\left[ \min_{r} \sum_{x}p_{X}(x) f_{\alpha, \text{PW}}(x, r) \right]^{\frac{1}{1-\alpha}}, & 0<\alpha<1
\end{cases} \\
&\overset{(g)}{=} \norm{p_{X}}_{\alpha}^{\frac{\alpha}{1-\alpha}} \overset{(h)}{=} \exp\{H_{\alpha}(X)\}, \\
&H_{\ln_{\alpha}t, \ln_{\alpha}t, \tilde{\ell}_{\alpha}}^{p_{X}}(X\mid Y) \notag \\ 
&= \min_{r_{X\mid Y}} \exp_{\alpha} \left\{\sum_{x,y}p_{X}(x)p_{Y\mid X}(y\mid x)\tilde{\ell}_{\alpha}(x, r_{X\mid Y}(\cdot \mid y)) \right\} \\ 
&= \exp_{\alpha} \left\{\min_{r_{X\mid Y}}\sum_{x,y}p_{X}(x)p_{Y\mid X}(y\mid x)\tilde{\ell}_{\alpha}(x, r_{X\mid Y}(\cdot \mid y)) \right\} \\ 
&= \exp_{\alpha} \left\{\sum_{y}p_{Y}(y)\min_{r_{X\mid Y}(\cdot\mid y)}\sum_{x}p_{X\mid Y}(x | y)\tilde{\ell}_{\alpha}(x, r_{X\mid Y}(\cdot | y)) \right\} \\ 
&\overset{(h)}{=} 
\begin{cases}
{\displaystyle \left[\sum_{y}p_{Y}(y) \max_{r_{X\mid Y}(\cdot\mid y)} \sum_{x}p_{X}(x) f_{\alpha, \text{PW}}(x, r_{X\mid Y}(\cdot \mid y)) \right]^{\frac{1}{1-\alpha}}}, \notag \\
\qquad \qquad \qquad \qquad \qquad \qquad \qquad \qquad \qquad \quad \alpha > 1, \\ 
{\displaystyle \left[\sum_{y}p_{Y}(y) \min_{r_{X\mid Y}(\cdot\mid y)} \sum_{x}p_{X}(x) f_{\alpha, \text{PW}}(x, r_{X\mid Y}(\cdot \mid y)) \right]^{\frac{1}{1-\alpha}}}, \notag \\ 
\qquad \qquad \qquad \qquad \qquad \qquad \qquad \qquad \qquad \quad 0<\alpha<1  
\end{cases}\\
&\overset{(i)}{=} \left[\sum_{y}p_{Y}(y)\norm{p_{X\mid Y}(\cdot \mid y)}_{\alpha}^{\alpha}\right]^{\frac{1}{1-\alpha}} 
\overset{(j)}{=} \exp\{H_{\alpha}^{\text{H}}(X\mid Y)\}, 
\end{align}
where 
\begin{itemize}
\item $(g)$ and $(i)$ follow from Proposition \ref{prop:power_score}.
\item $(h)$ and $(j)$ follow from Eqs.~\eqref{eq:Renyi_ent} and \eqref{eq:Hayashi_cond_renyi_ent}, respectively.
\end{itemize}

Finally, 
letting $\varphi(t):=\ln_{\frac{\alpha}{2\alpha-1}}t$ and $\psi(t):= \ln_{\frac{1}{\alpha}}t$, we obtain 
\begin{align}
\varphi\circ \psi^{-1}(t) &= \frac{2\alpha - 1}{\alpha-1} \left( \left[1 + \left( 1-\frac{1}{\alpha} \right)t\right]^{\frac{\alpha}{2\alpha-1}} - 1 \right), 
\end{align}
which leads to the following: 
\begin{align}
&V_{\varphi, \psi, g_{\text{$0$-$1$}}}^{p_{X}}(X\mid Y) \notag \\
&= \max_{r_{X\mid Y}} \left[\sum_{x}p_{X}(x)\left[\sum_{y}p_{Y\mid X}(y | x)r_{X\mid Y}(x | y)^{1-\frac{1}{\alpha}}\right]^{\frac{\alpha}{2\alpha-1}}\right]^{\frac{2\alpha-1}{\alpha-1}}. \label{eq:algebra_LP}
\end{align}
By replacing $p_{X}$ with $p_{X_{\frac{\alpha}{2\alpha-1}}}$ and from Eq.~\eqref{eq:Hayashi_cond_renyi_ent}, we obtain 
\begin{align}
V_{\varphi, \psi, g_{\text{$0$-$1$}}}^{p_{X_{\frac{\alpha}{2\alpha-1}}}}(X\mid Y) &= \exp\{-H_{\alpha}^{\text{LP}}(X\mid Y)\}.
\end{align}

\end{document}